\documentclass[aps,10pt,pra,showpacs,superscriptaddress,twocolumn,floatfix]{revtex4-1}

\usepackage{amsmath}
\usepackage{amsthm}
\usepackage{amssymb}
\usepackage{amsfonts}
\usepackage{bm} 
\usepackage[pdftex]{graphicx}
\usepackage{braket}
\usepackage{cases}
\usepackage[mathcal]{eucal}
\usepackage{bbold}

\newcommand{\ante}{ex-ante} 
\newcommand{\Ante}{Ex-ante} 
\newcommand{\antec}{{\ante} control} 
\newcommand{\Antec}{{\Ante} control} 
\newcommand{\post}{ex-post} 
\newcommand{\Post}{Ex-post} 
\newcommand{\postc}{{\post} control} 
\newcommand{\Postc}{{\Post} control} 
 
\newcommand{\qcq}{QCQ}

\newcommand{\etal}{\it et al.}

\DeclareMathOperator{\Tr}{Tr}
\DeclareMathOperator{\tr}{Tr} 
\DeclareMathOperator{\trhs}{Tr_{HS}}

\newcommand{\ketbra}[1]{\ket{#1}\!\bra{#1}} 
\newcommand{\kb}{\ketbra}
 
\makeatletter

\renewcommand{\le}{\leqslant}
\renewcommand{\ge}{\geqslant}
\newcommand{\C}{\mathbb{C}}

\renewcommand{\AA}{\mathcal{A}}
\newcommand{\CC}{\mathcal{C}}

\newcommand{\EE}{\mathcal{E}}
\newcommand{\FF}{\mathcal{F}}
\newcommand{\HH}{\mathcal{H}}
\newcommand{\II}{\mathcal{I}}

\newcommand{\LL}{\mathcal{L}}
\newcommand{\NN}{\mathcal{N}}
\newcommand{\OO}{\mathcal{O}}

\newcommand{\id}{\mathrm{id}}
\renewcommand{\a}{\alpha}
\renewcommand{\b}{\beta}
\newcommand{\g}{\gamma}
\renewcommand{\d}{\delta}
\newcommand{\ep}{\varepsilon}

\newcommand{\m}{\mu}
\newcommand{\n}{\nu}
\newcommand{\om}{\omega}
\newcommand{\Om}{\Omega}
\renewcommand{\r}{\rho}
\newcommand{\s}{\sigma}
\newcommand{\z}{\zeta}

\newcommand{\ve}{\bm}

\newcommand{\abs}[1]{{\left\vert #1 \right\vert}}

\newcommand{\paren}[1]{{\left( #1 \right)}}

\newcommand{\brac}[1]{{\left\{ #1 \right\}}}

\newcommand{\bracket}[1]{{\left[ #1 \right]}}
\newcommand{\brak}[1]{{\left[ #1 \right]}}

\newcommand{\f}{\frac}
\newcommand{\q}{\quad}

\newcommand{\da}[1]{#1^\dagger}
\newcommand{\norm}[1]{\left\|#1\right\|}
\newcommand{\dr}{discriminate and reprepare} 
\newcommand{\dn}{do nothing}
\newcommand{\gdn}{no measurement}

\newcommand{\ot}{\otimes}

\newcommand{\Ad}{\AA}

\newcommand{\Ref}{\eqref}
\newcommand{\nn}{\notag \\ }

\newcommand{\Vtr}[3]
{\bracket{\begin{array}{c} #1 \\ #2 \\ #3 \end{array}}}

\newtheorem{theorem}{Theorem}
\newtheorem*{theorem*}{Theorem}
\newtheorem{proposition}{Proposition}
\newtheorem*{claim*}{Claim}

\newtheorem*{corollary*}{Corollary}
\newtheorem{lemma}{Lemma}

\newcommand{\sbt}[1]{_{\text{#1}}}

\begin{document}
\title{
State protection by quantum control before and after noise}
\author{Hiroaki Wakamura}
\email{hwakamura@rk.phys.keio.ac.jp}
\affiliation{Department of Physics, Keio University, Yokohama 223-8522, Japan}
\author{Ry\^uitir\^o Kawakubo}
\email{rkawakub@rk.phys.keio.ac.jp}
\affiliation{Department of Physics, Keio University, Yokohama 223-8522, Japan}
\author{Tatsuhiko Koike}
\email{koike@phys.keio.ac.jp}
\affiliation{Department of Physics, Keio University, Yokohama 223-8522, Japan}
\affiliation{Research and Education Center for Natural Sciences, 
  Keio University, Yokohama 223-8521, Japan}
\date{February, 2017}

\begin{abstract}
We discuss the possibility of protecting the state of a quantum system
that goes through noise 
by measurements and operations before and after the noise process.
We extend our previous result 
on nonexistence of ``truly quantum'' protocols that 
protect an unknown qubit state against the depolarizing noise 
better than ``classical'' ones 
{}[Phys. Rev. A, 95, 022321 (2017)] 
in two directions. 
First, we show that the statement is also true in 
any finite-dimensional
Hilbert spaces, which was previously conjectured; 
the optimal protocol is 
either the {\dn} protocol or the {\dr} protocol, depending on the
strength of the noise. 
Second, in the case of a qubit, 
we show that essentially the 
same conclusion holds for any unital noise.
These results describe the fundamental limitations in 
quantum mechanics from the viewpoint of control theory. 
\end{abstract}
\maketitle

\section{Introduction} \label{sec:intro}

Quantum information technology, such as quantum computation, quantum cryptography, 
etc., is a new framework of information processing where quantum
states (e.g. qubits)  
bear information in place of classical bits. 
One of the difficulties in realization of those technologies is 
existence of noise. 
Since there is no isolated physical system in the world, 
the state inevitably undergoes 
noise processes 
caused by interactions with the environment.
The state evolution becomes irreversible and 
errors occur in information processing. 
In order to reduce the errors and make information processing feasible, 
protection of quantum states is an important
task~\cite{Sho95,Kni96,Reim05,Lid98,VioLlo98}. 

In the classical world, 
one can in principle protect any state against any noise, 
by taking the complete record of the state before the noise affects
the system.  
In the quantum world, it is not the case even 
if the state is not a probabilistic mixture. 
If one could do so,
then one could suppress
the disturbance caused by measurements 
and realize disturbance-free measurements. 
This would contradict with quantum 
measurement theory~\cite{DavisLewis,Ozawa84},
which implies that quantum measurements 
cannot extract the full information from a single sample 
and inevitably disturb the state. 
Thus impossibility of perfect state protection 
reflects the nature of quantum mechanics,
in the same way as 
impossibility of perfect 
state discrimination~\cite{helstrom, chefles98,unambig123,KawaKoi16} 
or quantum cloning~\cite{nocloning}.

Given this impossibility of perfect state protection, 
one may still want to consider control protocols which 
suppress 
the noise approximately. 
This is similar to pursuing 
the error-disturbance uncertainty 
relation~\cite{Ozawauniv,watanabe} 
or
theory of imperfect cloning~\cite{buzhil,werner98}.
Quantitative analysis of the limits 
in state protection  
may reveal 
the role played by measurements in state protection 
and whether there exists a comprehensive point of view 
to achieve the optimal state protection. 
That may 
also clarify 
the fundamental limitations 
in 
our ability to manipulate quantum systems 
and 
provide an 
operational characterization of the quantum world.

\begin{figure*}[t]
\centering 
\includegraphics[width=.7\linewidth]{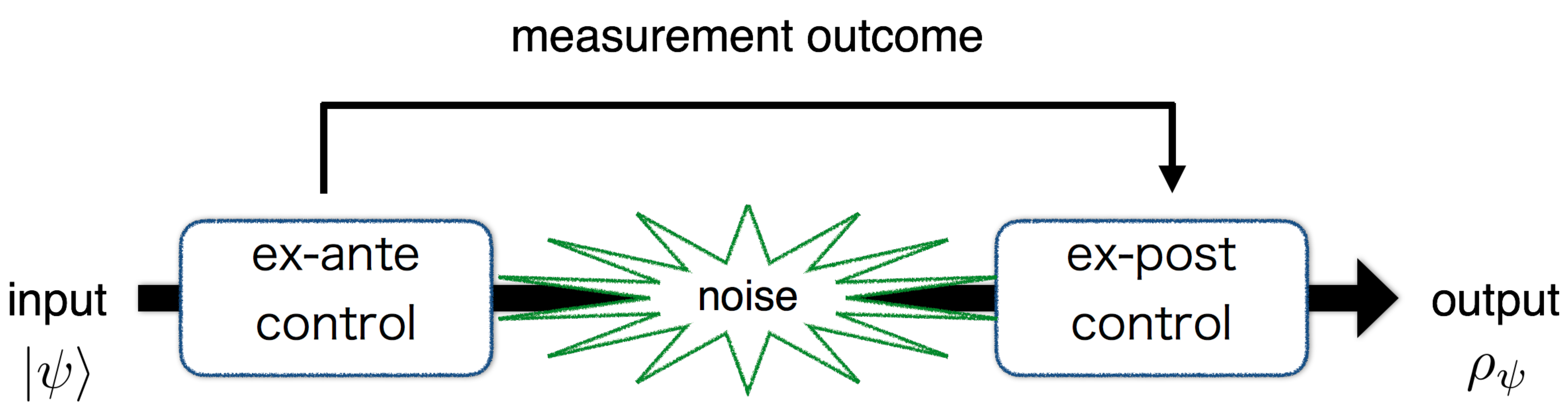}
\caption{The schematic diagram of the quantum control (\antec{} and \postc{}).}
\label{fig-scheme}
\end{figure*}

We would like to refer to recent works 
in the context of {\ante-\postc} scheme~\cite{paper1}.
The scheme consists of 
a general measurement before the noise process (\antec), 
and 
an operations after the noise process 
depending on the outcomes of the measurements (\postc), 
as depicted in Fig.~\ref{fig-scheme}.
The optimal protocol
which protects two states of a qubit 
solely by {\postc} 
has been derived by 
Branczyk {\etal}~\cite{bramen07} 
and 
Mendon\c{c}a {\etal}~\cite{mengil08}. 
An interesting interpretation of their results is that 
the optimal protocol 
detects the influence of the noise 
without discriminating the input states at all. 
Whether this strategy can be 
extended to other situations would be an intriguing
question. 
On the other hand, 
Zhang {\etal}~\cite{zhang08} showed 
that 
the {\postc} alone 
cannot protect 
a completely unknown pure state 
against 
the depolarizing noise. 
The present authors, in the qubit case, 
extended their results 
to general noise~\cite{paper1}: 
the optimal {\postc} protocol to protect a 
completely unknown pure state 
against an arbitrary noise
is a unitary 
operation, 
i.e., it is never beneficial to 
extract information in 
{\postc}. 
It is suggested 
by all these results 
that 
prior knowledge of the input pure states is essential 
to protect them. 
Thus, if one has no information on the input, 
one needs \antec{}. 
\Antec{} was considered 
by Korotokov and Keane~\cite{korotkov04} 
and then by 
Wang {\etal}~\cite{wang14}. 
Although their interests are the protocols with postselection, 
one can find in Ref.~\cite{wang14} 
some numerical results that 
suggest the existence of 
nontrivial \ante{}-\postc{} protocols (without postselection) 
which suppress the amplitude damping noise well. 
On the other hand, 
the present authors~\cite{paper1} proved
that 
there is no nontrivial \ante{}-\postc{} that can
suppress the depolarizing noise
better than ``classical''
protocols i.e. 
the ``{\dn}'' and the ``{\dr}'' protocols, 
where the latter consists of an {\ante} strong measurement and 
an {\post} repreparation of the state corresponding to the outcome of the
measurement.  
Given 
the suggested existence of nontrivial quantum 
protocol for the amplitude
damping noise 
and 
the nonexistence 
of such for the
depolarizing noise, 
it is natural to ask which 
class of noise allows or disallows  
nontrivial quantum state protection.

In this paper, 
we extend our previous results in Ref.~\cite{paper1},
on protection of a qubit against the depolarizing noise,
in two directions. 
We thereby partially solve 
the problem of protecting a completely unknown states against noise by 
\ante-\postc{}. 
The first direction is 
to extend the results to 
higher dimensional Hilbert spaces. 
To achieve that, 
we provide two observations that are powerful and quite general. 
One is convexity of the space of noise that are optimally
suppressed by the same protocol; 
the other is a sufficient condition on the noise 
for the {\dr} protocol to be optimal. 
As an application of these observations, 
we prove the conjecture in Ref.~\cite{paper1} that 
either the {\dn} or the {\dr} protocol is optimal to suppress the
depolarizing noise in general. 
The second direction is to widen the types of noise in the case of a
qubit. 
The class of noise considered is ``unbiased'' 
or 
{\em unital} noise, which leaves
the completely mixed state unchanged. 
The class 
contains  many types of noise appearing in quantum information 
including the depolarizing noise, but does not contain the amplitude damping noise. 
It can be said that unital noise makes any state more random
because it never decrease the (von Neumann) entropy of a quantum state.
We show that 
the optimal \ante{}-\postc{} protocol
to suppress unital noise 
is either 
a no measurement protocol
or the {\dr} protocol.

The paper is organized as follows.
After a very short review of 
basic mathematical tools 
in Sec.~\ref{sec:pre}, 
the state protection scheme by \ante-\postc{} is introduced 
in Sec.~\ref{sec:qucon}. 
We give general observations on noise suppression 
and show that the ``classical'' protocols suffice in state protection 
against the depolarizing noise 
in Sec.~\ref{sec-qcq}. 
Then we focus on the qubit case; 
we review the geometry of the space of unital noise 
in Sec.~\ref{sec:tetrahedron} 
and 
show that the ``classical'' protocols are optimal 
in Sec.~\ref{sec:main}. 
Sec.~\ref{sec:sum} is devoted to conclusion and discussions.

\section{Basics of quantum operations}
\label{sec:pre}

We shall introduce the basic mathematical tools and notation used in 
our analysis. 
Throughout the paper, we consider 
physical
systems which are
represented by 
a finite-dimensional 
Hilbert space. 
Let $\mathcal{H}$ be  such a Hilbert space and $\mathcal{L(H)}$ be the set of all linear
operators on $\mathcal{H}$.  
A quantum state is described by a density operator $\rho\in\LL(\HH)$ 
such that $\rho\ge0$ and $\tr \rho=1$. 
Since the control of quantum states consists of measurement and operations,
the mathematical map corresponding to measurement or operations is explained below.

Any physical evolution of a quantum state corresponds to 
a trace-preserving completely positive (TPCP) map, 
and vice versa (e.g.~\cite{Nielsen}).  
Here, a linear map $\mathcal{E}:\mathcal{L(H)}\rightarrow\mathcal{L(H)}$ is
said  {\it positive}\/ if $X\ge0$ implies
$\mathcal{E}(X)\ge0$ and
{\it completely positive (CP)}\/ if the map 
$\mathcal{E}\otimes {\id}_n$ 
is positive for every positive integer $n$, 
where $\id_n$ denotes the identity map on $\mathcal{L}(\C^n)=\C^{n\times n}$. 
The map $\mathcal{E}$ is
said 
{\it trace-preserving (TP)}\/ if $\tr\mathcal{E}(X)=\tr X$  for
any $X\in\mathcal{L(H)}$. 
 
Any physical measuring process corresponds to a
{CP instrument}, and vice versa~\cite{Ozawa84}. 
Here, a {\em CP instrument}\/ is 
a family 
$\{\II_\omega\}_{\omega\in\Omega}$ of CP maps 
with $\sum_{\omega\in\Omega}\II_\omega$ being trace-preserving. 
We assume that the set $\Omega$ of outcomes is finite throughout the
paper. 
The state evolution by the measurement is described as 
\begin{align}
 \rho \mapsto \frac{\II_\omega(\rho)}{\tr\II_\omega(\rho)},\quad
\text{with probability} \quad
\tr\II_\omega(\rho).
\end{align}
The probability distribution of measurement outcomes is 
described by a positive operator valued measure (\textit{POVM}), which is 
a family 
$\brac{M_\om}{}_{\om\in\Omega}$ 
of positive operators on $\HH$
such that 
$\sum_{\om\in\Omega} M_\om$ is the identity operator. 
A CP instrument $\brac{\II_\om}_{\om\in\Omega}$ defines a POVM 
$\brac{M_\om}{}_{\om\in\Omega}$ by 
$\Tr\rho M_\om=\Tr\II_\om(\rho)$ or
$M_\om=\II_\om^*(1)$, 
where an asterisk denotes the dual map. 
The dual map $\EE^\ast$ of $\EE\in\LL(\HH)$ is defined by 
$\tr\EE(X)Y=\tr X\EE^\ast(Y)$. 
We shall say that 
a POVM $\brac{M_\om}{}_{\om\in\Omega}$ 
and 
a CP instrument 
$\brac{\II_\om}{}_{\om\in\Omega}$ 
above 
are associated with each other. 
A POVM has all the information on the statistical properties of the 
measurement outcomes, while a CP instrument $\brac{\II_\om}$ 
has further information on the resulting states after the measurement. 

The space $\LL(\HH)$ of linear operators can be 
regarded 
as a Hilbert space with 
the Hilbert-Schmidt inner product
$\braket{X,Y}_{\mathrm{HS}}:=\Tr X^\dagger Y$. 
A linear map $\EE$ on $\LL(\HH)$ is interpreted as a linear operator 
on the Hilbert space $\mathcal{L(H)}$. 
The trace of such $\mathcal{E}$ 
is defined by 
\begin{align}
 \trhs\mathcal{E}:=\sum_{i} \braket{V_i,\mathcal{E}(V_i)}_{\mathrm{HS}},
\end{align}
where 
$\{V_i\}_i$ is an orthonormal basis of 
the Hilbert space $\mathcal{L(H)}$.
For example, when $\dim\HH=2$, 
an orthonormal basis of
$\LL(\HH)$ 
is given by 
$\{\sigma_\mu/\sqrt2\}_{\mu=0}^3$ 
where 
$\sigma_0$ is the identity operator 
and 
$\s_i$, $1\le i\le 3$, are the Pauli operators. 
Then the trace 
is written as 
\begin{align}
 \trhs\mathcal{E}
 =
 \f12
 \sum_{\mu=0}^3
 \tr\brak{\sigma_\mu\mathcal{E}(\sigma_\mu)}. 
 \label{eq-trhs-2d}
\end{align}

\section{The setup} 
\label{sec:qucon}

We shall discuss the noise suppression 
in the 
{\em \ante{}-\post{} quantum control scheme} below, 
which was proposed in Ref.~\cite{paper1}
(see Fig.~\ref{fig-scheme}). 
The scheme consists of the following: 
\begin{enumerate}
 \item State preparation: 
An unknown state 
is prepared. 
 \item \Antec: 
A measurement is performed, 
which is described by a CP instrument 
$\{\II_\om\}_{\om\in\Om}$. 
 \item Noise: 
The state undergoes an undesired evolution, called  ``noise,''
described by a TPCP map $\mathcal{N}$. 
 \item \Postc{}:
An operation, 
which depends on the
measurement outcome $\om$ of the \ante{} control, 
is performed on the system. 
This is described by a family 
$\{\CC_\om\}_{\om\in\Om}$ 
of TPCP maps.
\end{enumerate}
For given noise $\NN$, an 
\ante{}-\post{} control protocol is specified by 
the family $\{(\II_\om,\CC_\om)\}_{\om\in\Om}$. 
As we did in Ref.~\cite{paper1}, 
we focus on the case 
that the states prepared in Step 1 are pure 
and 
is 
{\em completely unknown}, i.e., 
the prior 
probability distribution is uniform on the unit sphere in
$\mathcal{H}$, 
though one can consider more general cases within the scheme above. 
Our problem is 
to find an {\em optimal}\/ \ante-\postc{} protocol 
$\brac{(\II_\om,C_\om)}_{\om\in\Omega}$
for given noise $\NN$
such that 
the states after the measurement with outcome $\om$ 
are as close to the original state $\ketbra\psi$ as
possible.  
We evaluate the closeness by 
{\em fidelity},
which is expressed by 
$F(\r,\kb{\psi}):=\braket{\psi|\r|\psi}$ 
if one of the two states is 
pure (e.g.~\cite{Nielsen}),
and the optimality is defined by 
the average 
fidelity 
with respect to the probability to obtain each outcome $\om$ 
and with respect to that of each input state $\ket\psi$.

An advantage of the choice is 
that 
the resulting averaged evaluation function, the {\it average fidelity} 
\begin{align}
 \bar{F}=\int_{\|\ket{\psi}\|=1} d\psi
 \braket{\psi|\EE(\ketbra\psi)|\psi},
 \label{eq:avefid1} 
\end{align}
depends on 
the {\em average operation}\/ 
\begin{align}
\EE:=
 \sum_{\om=1}^M\CC_\om\circ\mathcal{N}\circ\II_\om 
 \label{eq-averaged-op}
\end{align}
which is a TPCP map. 
We will use the formula~\cite{paper1} 
\begin{align}
 \bar{F}=
 \frac{d+\trhs\EE}{d(d+1)}, 
\label{eq:avefid2} 
\end{align}
where 
$d:=\dim \mathcal{H}$ 
and 
$\EE$ is 
the average operation 
\eqref{eq-averaged-op} 
of the protocol $\brac{(\II_\om,\CC_\om)}_\om$.

\section{results in general system} 
\label{sec-qcq}
In this section, 
we consider noise in the system of a 
$d$-dimensional Hilbert space $\HH$. 
We first give the important property 
which comes from 
convexity of the space of noise (TPCP maps) 
(Proposition~\ref{lem-convexity}). 
Second, we give a sufficient condition on the noise 
for the {\dr} protocol to be optimal 
(Proposition~\ref{prop-dr}); 
this solves 
the difficulties 
in seeking optimal \ante-\postc{} protocol 
for a wide range of noise. 
Third, we combine these facts to 
solve the optimality problem for the depolarizing noise, 
the case $d\ge3$ of which 
was a conjecture in our previous work~\cite{paper1}. 
This serves as a demonstration of the strength of 
Propositions~\ref{lem-convexity} and 
\ref{prop-dr}. 
\begin{proposition}
  The space of all noise processes $\NN$ 
  that are optimally suppressed by a 
  single {\ante-\postc} protocol 
  $\brac{(\II_{\om},\CC_{\om})}_{\om\in\Om}$ 
  is convex. 
  \label{lem-convexity}
\end{proposition}
\begin{proof}
  The claim is equivalent to the following: 
  if a control protocol $\brac{(\II_{\om},\CC_{\om})}$
  optimally suppresses two noise processes $\NN_1,\NN_2$ then 
  it also optimally suppresses any mixture of them,
  $\NN=(1-\a)\NN_1+\a\NN_2$, $0\le\a\le1$. 
  Let $\brac{(\II_{\om},\CC_{\om})}$ be 
  such a protocol. 
  By 
  \Ref{eq-averaged-op} and 
  \Ref{eq:avefid2}, 
  the optimal protocol 
  maximizes $\trhs \EE$. 
  For any protocol $\brac{(\II'_\om,\CC'_\om)}$, 
  one has
  \begin{align}
    &
    \trhs \II'_\om\circ \NN\circ \CC'_\om
    \nn
    &
    =
    (1-\a)
    \trhs \II'_\om\circ \NN_1\circ \CC'_\om
    +\a
    \trhs \II'_\om\circ \NN_2\circ \CC'_\om
    \nn
    &\le
    (1-\a)
    \trhs \II_\om\circ \NN_1\circ \CC_\om
    +
    \a
    \trhs \II_\om\circ \NN_2\circ \CC_\om
    \nn
    &=
    \trhs \II_\om\circ \NN\circ \CC_\om. 
  \end{align}
  Thus 
  $\brac{(\II_\om,\CC_\om)}$ maximizes 
  $
  \trhs \II'_\om\circ \NN\circ \CC'_\om
  $ hence 
  it optimally suppresses $\NN$. 
\end{proof}

Let us consider operations of the form 
\begin{align}
 \FF(\r)=\sum_k \r_k\tr M_k\r, 
\label{eq-ebtp}
\end{align}
in which 
one measures
the input state $\r$
by a POVM $\brac{M_k}$ and prepares a state
$\r_k$ according to the measurement outcome $k$. 
Such an operation is sometimes called an 
quantum-classical-quantum (\qcq)
channel~\cite{holevo}. 

\begin{proposition}
  \label{prop-dr}
  In the scheme of {\ante-\postc}, 
  any {\qcq} noise is optimally suppressed by 
  the {\dr} protocol 
  $\brac{(\II_\om,\CC_\om)}_{\om=1,...,d}$ 
  defined by 
   \begin{align}
     \II_\om(\rho)
     &=
     \ketbra{\phi_\om}\rho\ketbra{\phi_\om},
     \label{eq-DR-I-gendim}
     \\
     \mathcal{C}_\om(\rho)
     &=
     \ket{\phi_\om}\bra{\phi_\om}{\tr\r},
     \label{eq-DR-C-gendim}
   \end{align}
   where $\brac{\ket{\phi_\om}}_{\om=1,...,d}$ 
   is an arbitrary orthonormal basis of $\mathcal{H}$. 
   The optimal average fidelity is  
   \begin{align}
     \bar F_{\mathrm{DR}}=\f2{d+1}. 
     \label{eq-Fbar-DR-gendim}
   \end{align}
\end{proposition}

We give a remark before proving the proposition. 
The {\dr} protocol above 
is to discriminate the input state between certain 
$d$
orthogonal states $\brac{\ket{\phi_\om}}_{\om=1,...,d}$ 
and reprepare the discriminated state $\ket{\phi_\om}$ after the noise. 
The value 
of the average fidelity is 
\Ref{eq-Fbar-DR-gendim},  
which is independent from $\NN$ and 
from the choice of $\brac{\ket{\phi_\om}}$. 
Indeed, 
it follows from \eqref{eq-DR-C-gendim} 
that 
$
\CC_\om\circ\NN\circ\II_\om=
\CC_\om\circ\II_\om
$ holds 
for any trace-preserving $\NN$. 
From Eq.~\Ref{eq-trhs-2d}, one has 
\begin{align}
&\trhs 
\CC_\om\circ\II_\om
=\f12\Tr\II_\om\circ\CC_\om(1)
  =1. 
  \label{eq-DR-trhs}
\end{align}
Then the average fidelity \Ref{eq-Fbar-DR-gendim} is 
obtained by the general formula~\eqref{eq:avefid2} for $\bar F$.

\begin{proof} 
[Proof of Proposition~\ref{prop-dr}]
First, 
we observe 
that 
if $\NN$ is a {\qcq} channel, 
so is the average operation $\EE=\sum_\om\CC_\om\circ\NN\circ\II_\om$ 
for any protocol $\brac{(C_\om,\II_\om)}$. 
This is so because if $\NN$ is written in the form 
\Ref{eq-ebtp}, then one has 
\begin{align}
  \EE(\r)=\sum_{\om,k}\CC_\om(\r_k)
  \Tr\brak{ \II_\om^*(M_k)\r }, 
\end{align}
with each $\CC_\om(\r_k)$ being a state and 
$\brac{ \II_\om^*(M_k)}$ being a POVM. 
Second, it is shown in Ref.~\cite{BruMac99} that 
the maximum 
average fidelity 
between the input and output states 
for {\qcq} channels $\FF$ is given by 
$\bar{F}=2/(d+1)$. 
Therefore,  
the average fidelity $\bar F$ 
for 
any protocol 
$\brac{(C_\om,\II_\om)}$ 
does not exceed 
that value. 
On the other hand, 
the value can 
be attained by the {\dr} protocol in the theorem. 
\end{proof}

The proposition above partially solve the problem of 
state protection by {\ante-\postc}; 
if the noise turns out to be {\qcq}, then the {\dr} protocol is 
optimal. 
Several equivalent conditions for a map to be {\qcq} is
known~\cite{HorShoRus03}. 
By using one of such conditions, 
we can prove 
the conjecture proposed in Ref.~\cite{paper1} as a corollary 
of Proposition~\ref{prop-dr}. 
\begin{theorem}
The optimal {\ante}-{\post} 
protocol 
 $\{(\II_\om,\CC_\om)\}_{\om=1,...,d}$
for the depolarizing noise 
\begin{align}
  \NN
 =
 (1-\varepsilon)\rho+\varepsilon\frac{1}{d}\tr
 \rho, 
 \label{depo-noise}
\end{align}
 is given as follows.

(i) When the noise is weak, $\ep\le d/(d+1)$, 
   the {\dn} protocol 
   $\{(\II_\om,\CC_\om)\}_{\om=1}$ with 
\begin{align}
  (\II_1,\CC_1)=(\id, \id) 
\end{align}
is optimal. 
The optimal average fidelity is 
$\bar F_{\mathrm{DN}}
=1-\ep(d-1)/d
$. 

(ii) When the noise is strong, $\ep\ge d/(d+1)$, 
the {\dr}  protocol 
$\{(\II_\om,\CC_\om)\}_{1\le \om\le d}$ 
given by 
\Ref{eq-DR-I-gendim} and 
\Ref{eq-DR-C-gendim}
is optimal. 
The optimal average fidelity is 
$\bar F_{\mathrm{DR}}=2/(d+1)$. 
\label{conj:kekka3}
\end{theorem}
\begin{proof}
  We first show that the 
  noise $\NN$ is {\qcq} if and only if $\ep\ge d/(d+1)$. 
  It is known~\cite{HorShoRus03}  
  that a linear map $\NN$ is {\qcq} if and only if 
  the image of the maximally entangled state 
  $\ket\Psi:=(1/\sqrt d)\sum_{i}\ket{ii}$ by $\id\ot\EE$ is 
  separable. 
  From \Ref{depo-noise}, one has
  \begin{align}
    (\id\ot\NN)\paren{\kb\Psi}=(1-\ep)\kb\Psi+\ep\f1d. 
  \end{align}
  The right hand side above is 
  a mixture of the maximally entangled state and 
  the completely mixed state. 
  The condition that such a state 
  is separable, 
  hence $\EE$ is {\qcq}, 
  is
  $\ep\ge d/(d+1)$~\cite[Sec.~IVB]{VidTar99}. It then follows from
  Proposition~\ref{prop-dr} that the {\dr} protocol is optimal when 
  $\ep\ge d/(d+1)$. 
  On the other hand, 
  the {\dn} protocol gives the average fidelity 
  \begin{align}
    \bar F\sbt{DN}=1-\ep+\f\ep d, 
    \label{eq-Fbar-DN-gendim}
  \end{align}
  which follows from 
  \Ref{eq:avefid2}
  and 
  \Ref{depo-noise}. 
  When $\ep=d/(d+1)$, 
  one has $\bar F\sbt{DN}=\bar F\sbt{DR}=2/(d+1)$ 
  so that 
  the noise $\NN$ is optimally
  suppressed 
  both 
  by the {\dr} and {\dn} protocols. 
  Furthermore, when $\ep=0$, the noise $\NN=\id$ 
  is optimally 
  suppressed by the {\dn} protocol. 
  Therefore, by Proposition~\ref{lem-convexity}, any noise $\NN$ with 
  $0\le \ep\le d/(d+1)$, a convex combination of the two cases above, 
  is optimally suppressed by the {\dn} protocol, when the optimal
  average fidelity is given by \Ref{eq-Fbar-DN-gendim}. 
\end{proof}

\begin{figure*}[t]
\centering 
\includegraphics[width=.8\linewidth]{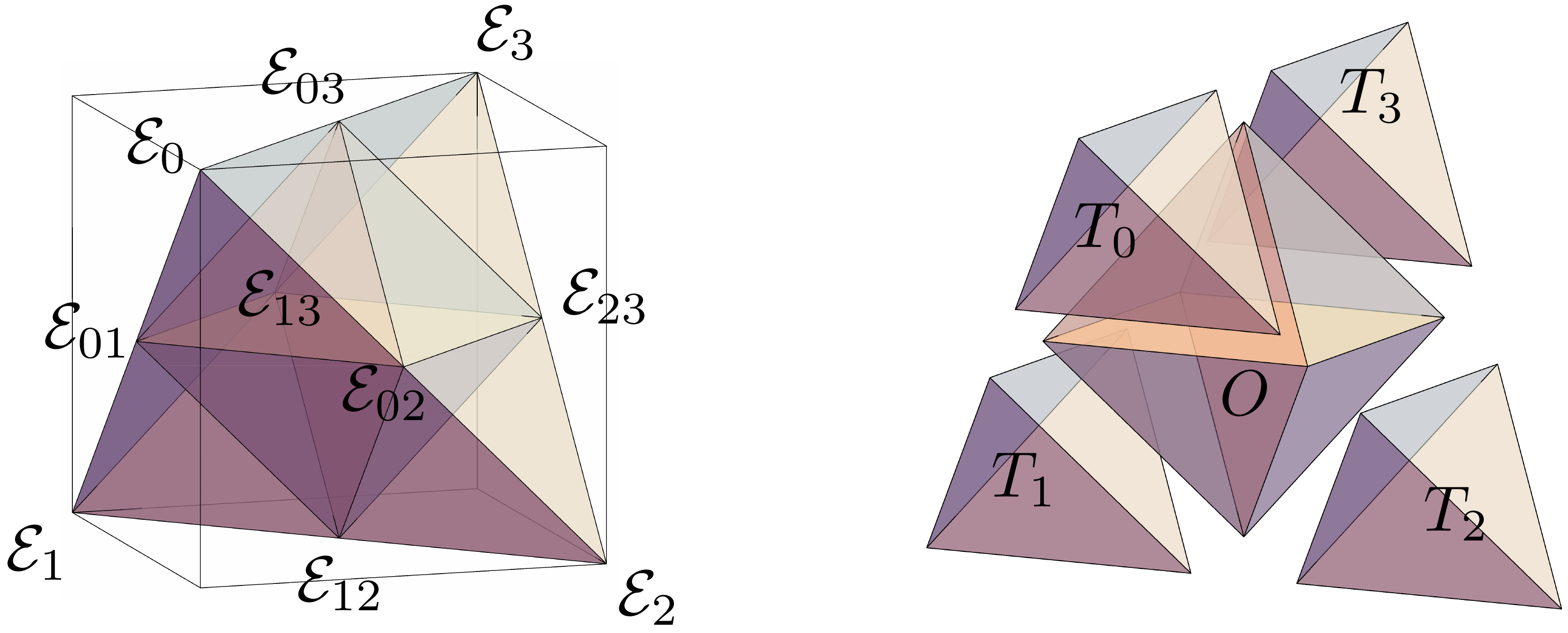}
\caption{Left: 
  The space of unital noise, 
  up to input and output unitary operations, 
  form a tetrahedron $T$ whose vertices are 
  $\EE_0 (1,1,1)$, 
  $\EE_1 (1,-1,-1)$ 
  $\EE_2 (-1,1,-1)$ and 
  $\EE_3 (-1,-1,1)$. 
  Unital noise is a convex combination of the 
  unitary operations $\EE_\m$. 
  Right: 
  The tetrahedron $T$ is decomposed into an octahedron $O$ 
  and four tetrahedra $T_\m$. 
  The six vertices of 
  $O$ 
  are the midpoints $\EE_{\m\n}$ of the edges, 
  which are $(\pm1,0,0)$, $(0,\pm1,0)$, and $(0,0,\pm1)$. 
  This decomposition will be important in 
  Theorem~\ref{theo:kekka2}. 
}
\label{fig-tetra} 
\end{figure*}

\section{Geometry of unital noise} 
\label{sec:tetrahedron}
To discuss protection of the state of a qubit 
against unital noise in the next section, 
we briefly introduce the geometry of unital TPCP maps. 

A linear map $\EE$ 
on 
operators 
is said {\em unital}\/ 
if it preserves the identity operator, $\EE(1)=1$. 
The class of unital noise appears commonly 
in quantum information. 
We can interpret unital TPCP maps as ``unbiased,'' 
because it keeps the completely mixed state. 
An important characteristic of a unital TPCP map $\EE$ 
is that it never decreases the von Neumann entropy $S(\r)$ 
of a quantum state $\r$, i.e., $S(\EE(\r))\ge S(\r)$. 
Thus one can say that unital noise $\EE$ 
always increases (or at least keeps) 
the 
randomness 
of the input state $\r$. 
We remark that one way to understand the inequality above is 
to apply the well-known nonincreasing property 
of quantum relative entropy 
$S(\r||\s):=\Tr \brak{ \r \ln\r-\r \ln\s }$ 
under a TPCP map $\EE$, 
i.e. $S(\EE(\r)||\EE(\s))\le S(\r||\s)$, to the state  
$\s=1/d$. 
In this section, 
we briefly summarize the facts about the convex structure of the
space of unital TPCP maps on a qubit. 

We consider the set of unital TPCP maps on $\LL(\HH)$. 
From the definition of unital TPCP maps, it is easy to see that the
set is convex in $\LL(\LL(\HH))$, is closed under composition, and
contains all unitary operations. 
If $\dim \HH=2$, 
because each Pauli operator $\sigma_\mu$ is unitary, 
the map
\begin{align}
  \NN
  =
  \sum_{\m=0}^3\a^{\m}\Ad_{\s_\m}, 
  \q
 \a^\m\ge0, \q \sum_\m\a^\m=1,
 \label{eq-nnalpha}
\end{align}
is unital and TPCP, 
where $\Ad_U(\r):=U\r\da U$. 
It follows that the map
\begin{align}
 \NN'=\Ad_V\circ\,\NN\circ \Ad_U 
  \label{eq-nnuv}
\end{align}
is also unital and TPCP 
if $U$ and $V$ are unitary operators.
Conversely, it is known~\cite{KingRus} 
that the above $\NN'$ runs over all unital TPCP maps
when we vary $\alpha^\m$, $U$ and $V$.
Thus, apart from the degree of freedom of fixed unitary operations on the
input and output states, 
the unital TPCP maps is parameterized by $\alpha^\m$ 
as in \eqref{eq-nnalpha}. 
They form a tetrahedron $T$ with vertices at 
$\a^\m=(1,0,0,0)$,
$(0,1,0,0)$,
$(0,0,1,0)$ and
$(0,0,0,1)$.
These vertices correspond to unitary operations $\Ad_{\sigma_\mu}$.

For later calculation,
we introduce a new coordinate system $(d^i)$ 
by
\begin{align}
{\scriptsize \Vtr{d^1}{d^2}{d^3}}
=
  \a^0{\scriptsize \Vtr111 }
  +\a^1{\scriptsize \Vtr1{-1}{-1} }
  +\a^2{\scriptsize \Vtr{-1}1{-1} }
  +\a^3{\scriptsize \Vtr{-1}{-1}1 },
\end{align}
 so that 
\begin{align}
  \NN(\s_i)
  =
  d^{i}\s_i 
  \q\text{(no sum)} 
  \label{eq-unital-canonical}
\end{align}
holds.
In the coordinate system $(d^i)$,
the tetrahedron $T$ has the vertices at 
$\EE_0(1,1,1)$, 
$\EE_1(1,-1,-1)$, 
$\EE_2(-1,1,-1)$, 
and 
$\EE_3(-1,-1,1)$ (Fig.~\ref{fig-tetra}). 
Let $\EE_{\m\n}$ $(\m\neq\n)$ be the midpoint of $\EE_\m$ and $\EE_\n$, and 
let $O$ be the octahedron whose vertices are 
the six midpoints $\EE_{\m\n}$. 
Then the space $\overline{T\setminus O}$ consists of four smaller 
tetrahedra. 
Let $T_\m$ ($\m=0,1,2,3$) be each of such tetrahedra that contains $\EE_\m$.
Thus $T=O\cup T_0\cup T_1\cup T_2\cup T_3$.
In the following, we identify each unital noise represented by
\eqref{eq-unital-canonical} [or \eqref{eq-nnalpha}] and a point in the
tetrahedron $T$.

We remark on the tetrahedral symmetry of $T$, which is 
the remaining symmetry on $T$ caused by the freedom of $U$ and $V$ in 
\Ref{eq-nnuv}. 
A pair $(U,V)$ of unitary operators determines by \eqref{eq-nnuv} 
an automorphism $\NN\mapsto \NN'$ of the convex space of unital TPCP maps.
If the pair $(U,V)$ is properly chosen, 
the automorphism sends the tetrahedron $T$ to itself so that 
it is a tetrahedral symmetry map. 
For example, 
when $(U,V)=(1,\sigma_3)$, 
the automorphism is 
$
(d^1,d^2,d^3)\mapsto 
(-d^1,-d^2,d^3)
$ 
and 
sends $(\EE_0,\EE_1,\EE_2,\EE_3)$ 
to $(\EE_3,\EE_2,\EE_1,\EE_0)$.
When 
$(U,V)=(e^{i\pi \s_3/4},e^{-i\pi \s_3/4})$, 
the automorphism is 
$
(d^1,d^2,d^3)\mapsto 
(d^2,d^1,d^3)
$ and 
sends $(\EE_0,\EE_1,\EE_2,\EE_3)$ to $(\EE_0,\EE_2,\EE_1,\EE_3)$. 
Thus, the six pairs $(U,V)=(1,\sigma_i)$ and
$(e^{i\pi \s_i/4},e^{-i\pi \s_i/4})$, 
$1\le i\le3$, generate the tetrahedral symmetry group 
consisting of $4!$ maps (all permutations of the indices). 
In particular, four small tetrahedrons $T_\mu$ are equivalent 
if we disregard unitary operations before and after the noise.

\section{Qubit under unital noise} 
\label{sec:main}

Now we present our main result for 
state protection against 
arbitrary unital noise 
when
$\dim\HH=2$. 
The theorem below generalize the result 
for the depolarizing noise~\cite{paper1}  
to general unital noise. 

In general, there is a trade-off between 
the information gained and 
the disturbance caused by the ex-ante control.
Though one might expect that a protocol 
with weak \ante{} measurements and weak 
\postc{} would be optimal, 
the theorem states that this is not the case. 

Below we discuss control protocols $\brac{\paren{\II_\om,\CC_\om}}$ for noise $\NN$ of the  
form~\eqref{eq-nnalpha} without loss of generality~\cite{unitary}. 

\begin{theorem}
  Let $\dim\HH=2$ and 
  let $\NN\in T$. Thus $\NN$ is 
  unital noise of the form \eqref{eq-unital-canonical}. 
 Then the optimal \ante{}-\postc{} protocol
 $\brac{(\II_\om,\CC_\om)}_{\om=1,2,...}$  
 suppressing noise $\NN$ is given as follows.

   (i) When $\NN\in T_\m$, 
   the optimal protocol 
   $\brac{(\II_\om,\CC_\om)}_{\om=1}$ 
   is a {\gdn} protocol defined by 
   \begin{align}
     \II_1= \id,\quad \CC_1=\Ad_{\s_\m}. 
   \end{align}
   The optimal average fidelity is 
   \begin{align}
     \bar F_{\text{\rm \gdn{}}}
     =
      \f12+\f
      { \abs{ d^1 } + \abs{ d^2 } + \abs{ d^3 } } 6.
   \end{align}

   (ii) 
   When $\NN\in O$, 
   the {\dr} protocol 
   $\brac{(\II_\om,\CC_\om)}_{\om=1,2}$, 
   defined by 
   \Ref{eq-DR-I-gendim} and \Ref{eq-DR-C-gendim} with $d=2$, 
   is optimal. 
   The optimal average fidelity is  
   $\bar F_{\mathrm{DR}}=2/3$. 
 \label{theo:kekka2}
\end{theorem}

The \gdn{} protocol above
does not involve any measurement 
and merely cancels the reversible part of the noise.
When $\m=0$, it is nothing but 
the \dn{} protocol and 
the value of the average fidelity 
$\bar F_{\textrm{DN}}$ 
can be obtained by direct calculation: 
\begin{align}
 \bar F_{\textrm{DN}} 
 &= \int_{\|\ve x\|=1}d\nu 
 \Tr\brak{\f{1+\ve x\cdot\ve\s}2 
   \f{1+\sum_id^ix^i\s_i}2}
 \nn
 &=\f12+\f{\sum_id^i}{6}, 
\label{eq-attain-dn}
\end{align}
where 
$\nu$ is the normalized uniform measure on a unit sphere. 
The {\dr} protocol appeared 
in Proposition~\ref{prop-dr} 
and $\bar F\sbt{DR}=2/3$ follows from Eq.~\Ref{eq-Fbar-DR-gendim} and 
$\dim\HH=2$. 
The {\gdn} and {\dr } protocols are considered ``classical'' because 
one either
performs no
quantum measurement at all
or only uses the classical information
extracted by the \ante{} measurement. 
 
The difference between 
noise in $T_\mu$ and $O$
can be understood as the strength of noise.
In fact, the vertices $\EE_\mu$ 
are unitary operations, while
the origin, which always outputs the completely mixed state,
entirely destroys the initial state.
The theorem above
states that 
the optimal protocol 
depends on the strength of the noise
and suddenly changes at the threshold,
with no intermediate regime
in which truly quantum protocols are optimal.

\begin{lemma}
  In the two-dimensional Hilbert space $\HH$, 
  consider
  the noise $\NN=\EE_{0i}$, $1\le i\le 3$, i.e., 
  \begin{align}
    \NN(\r)=\f12\paren{ \r + \s_i\r\s_i  }. 
    \label{eq-strong-depha-noise}
  \end{align}
  The {\dn} 
  and the {\dr} protocols
  are optimal 
  \ante{}-\postc{} protocols to supress $\NN=\EE_{0i}$. 
  The optimal average fidelity is  
  $\bar F=2/3$. 
  \label{prop-depha}
\end{lemma}
\begin{proof}
  We give a proof for $i=3$; 
  the cases $i=1,2$ are essentially the same. 
Let us show that 
the dephasing noise $\NN=\EE_{03}$ 
is a {\qcq} channel. 
Let 
$P_0$ and $P_1$ be the projection to the eigenspaces of $\s_3$ with
eigenvalues $1$ and $-1$. 
Then, inserting $\s_0=P_0+P_1$ and
$\s_3=P_0-P_1$ 
to 
$\NN(\r)=\f12\paren{ \s_0\r\s_0 + \s_3\r\s_3 }$, 
one can rewrite 
$\NN$ in the form \Ref{eq-ebtp} with 
$\r_k:=P_k$ and $M_k:=P_k$, $k=0,1$. 
Then the optimality of the {\dr} protocol in suppressing $\NN$ 
follows 
from Proposition~\ref{prop-dr}. 
On the other hand, 
the value $\bar F=2/3$ can 
be attained also by the {\dn} protocol, 
which can be seen by substituting 
$(d^i)=(0,0,1)$ 
into $\bar F\sbt{DN}$ in \Ref{eq-attain-dn}. 
Therefore the claim is true. 
\end{proof}
We give in the Appendix an alternative 
proof of Lemma~\ref{prop-depha} 
which does not depend on Proposition~\ref{prop-dr} 
and is based on a direct calculation. 
Now, let us prove Theorem~\ref{theo:kekka2}. 
\begin{proof}[Proof of Theorem~\ref{theo:kekka2}]
  (i) When $\NN\in T_0$, 
  it is trivial that 
  the {\dn} protocol
  $\brac{(\II_\om,\CC_\om)}_{\om=1}=\brac{(\id,\id)}$
  optimally suppresses the noise
  $\EE_0=\id$ 
  with $\bar F=1$. 
  From Lemma~\ref{prop-depha}, 
  this protocol 
  also suppresses optimally 
  the noise $\EE_{0i}$, $1\le i\le3$. 
  Then from Lemma~\ref{lem-convexity}, 
  the {\dn} protocol 
  optimally suppresses 
  any noise $\NN$ in the convex hull $T_0$ 
  of 
  $\EE_0$, $\EE_{01}$, $\EE_{02}$, and $\EE_{03}$. 
  The average fidelity $\bar F$ is given by Eq.~\eqref{eq-attain-dn}. 
  
  When $\NN\in T_i$,  
  the noise $\Ad_{\s_i}\circ\NN$ is in $T_0$, 
  as explained in the preceding section,
  and hence optimally suppressed by the {\dn} protocol $\brac{(\id,\id)}$.
  Therefore 
  $\NN$ here is optimally suppressed by 
  $\brac{(\id, \Ad_{\s_i})}$~\cite{unitary},
  which is a {\gdn} protocol. 
  The average fidelity $\bar F$ is given by Eq.~\eqref{eq-attain-dn}
  with $d^j$ $(j\ne i)$ replaced with $-d^j$.

  (ii) By Lemma~\ref{prop-depha}, three of the vertices 
  $\EE_{0i}$ of the octahedron $\OO$ are 
  optimally 
  suppressed by the {\dr} protocol 
  $\brac{(\II_\om,\CC_\om)}$ 
  defined by \Ref{eq-DR-I-gendim} and \Ref{eq-DR-C-gendim}. 
  The other vertices can be flipped to one of the former three by
  $\Ad_{\s_i}$, as explained 
in the preceding section. 
  Thus they are 
  optimally 
  suppressed by the {\dr} protocol 
  $\brac{(\II_\om,\CC_\om\circ\Ad_{\s_i})}$~\cite{unitary}. 
  Furthermore, 
  the protocol $\brac{(\II_\om,\CC_\om)}$ 
  gives the same average fidelity as  
  $\brac{(\II_\om,\CC_\om\circ\Ad_{\s_i})}$, 
  because the average operation~\Ref{eq-averaged-op} yields
  $\EE=\sum_\om \CC_\om\circ\II_\om$ for the both protocols.
  Thus all vertices of $\OO$ are 
  optimally 
  suppressed by the same {\dr} protocol 
  $\brac{(\II_\om,\CC_\om)}$. 
  Recalling that 
  $O$ is the convex hull of these six vertices, 
  one concludes, by Proposition~\ref{lem-convexity}, 
  that 
  each $\NN\in\OO$ is 
  optimally 
  suppressed by the {\dr} protocol 
  $\brac{(\II_\om,\CC_\om)}$. 
\end{proof}

\section{Conclusion and discussions} \label{sec:sum}

We discussed the problem of 
protecting a completely unknown
state against given unital noise by \ante{} and \postc{}
scheme.  
A protocol in the scheme is 
described mathematically by 
a family of pairs, 
$\brac{(\II_\om,\CC_\om)}_{\om\in\Omega}$, 
where $\brac{ \II_\om }_{\om\in\Om}$ is 
the CP instrument with the set $\Omega$ of outcomes 
which describes 
the {\ante} measurement and 
the map $\CC_\om$ 
is the TPCP map 
which describes the {\post} operation when the outcome $\om$ is obtained. 
To evaluate the closeness of the input and output states, 
we have chosen the average fidelity $\bar F$ between the input and output states. 

We presented two general observations on 
convexity of the noise that are optimally suppressed by the same
protocol (Proposition~\ref{lem-convexity}) 
and 
a sufficient condition for the {\dr} protocol to be optimal 
(Proposition~\ref{prop-dr}). 
These observations 
enabled us to prove the previous conjecture as 
Theorem~\ref{conj:kekka3},
which states 
that the depolarizing noise is optimally suppressed 
by 
classical protocols; namely, 
the {\dn} protocol is optimal 
if the noise is weak 
and 
the {\dr} protocol is optimal 
if the noise is strong. 
Then we focused on the 
case of a qubit system 
and generalized the result to the class of unital 
noise, 
which can be considered as unbiased because it preserves 
the completely mixed state. 
We proved that arbitrary unital noise 
is optimally suppressed by 
the classical protocols, the no measurement protocol or 
the discriminate and reprepare protocol depending on the strength of the noise (Theorem~\ref{theo:kekka2}).

Our results suggest that 
one can perform nontrivial suppression
of noise only by taking advantage of the bias of the noise. 
This gives a natural understanding for the previously known facts and 
numerical evidences that 
nontrivial suppression is possible against the amplitude damping noise but
is not possible against the depolarizing noise~\cite{paper1,wang14}. 
For a deeper and more precise understanding
of state protection in this direction,
it will be necessary to examine our 
Theorem~\ref{theo:kekka2} in higher dimensional Hilbert spaces 
and  
to investigate optimal {\ante-\postc} protocols against 
non-unital noise. 

We would like to emphasize that 
the new method 
based on Propositions~\ref{lem-convexity} and
\ref{prop-dr} 
is much more general than the previous 
one~\cite{paper1} 
which involved 
a detailed 
estimation of a function of several variables.
Thus it may give not only a 
systematic approach 
but also a general 
perspective to the problem of state protection. 
The proofs of Theorems~\ref{conj:kekka3} and \ref{theo:kekka2} 
were to find several 
TPCP maps which are optimally suppressed by a single
protocol and 
to derive the optimality of the protocol in their convex hull.  
It was especially important 
to find the particular 
noise (TPCP map) that is optimally suppressed by 
two or more different protocols simultaneously,
such as $\EE_{\m\n}$ in the proof of Theorem~\ref{theo:kekka2}. 
In other words, it is essential 
to find the {\em watersheds}\/ ({\em critical points}) 
in the space of noise 
for determination of the 
{\em basin of optimality}\/ ({\em convex domain}) 
of a protocol. 
Further applications of the method may reveal the 
``phase diagram'' of optimality 
in the space of noise. 
We hope that this work serves as a prototype for such developments. 

\section*{Acknowledgments}
T.~K. acknowledges the support from MEXT-Supported Program for the
Strategic Research Foundation at Private Universities ``Topological
Science'' and from Keio University Creativity Initiative ``Quantum
Community.''

\appendix*
\section{An elementary proof of Lemma~\ref{prop-depha}}
  We give a proof for the case $i=3$, the dephasing noise.
  The cases $i=1,2$ are similar. 
  From Eq.~\eqref{eq:avefid2}, 
  the optimal protocol 
  $\{(\II_\om,\CC_\om)\}_{\om=1,2,...}$ 
  is the maximizer of 
  $\sum_{\om}f_\om$, with 
  \begin{align}
    f_\om
    :=& \, \trhs\CC_\om\circ\NN\circ\II_\om
    = \trhs\NN\circ\II_\om \circ \CC_\om
    \nn
    =& \, 
    \f12\tr \II_\om\circ \CC_\om(1)
    +
    \f12\tr \s_z\II_\om \circ \CC_\om(\s_z)
    \nn
    =& \, 
    \f12\tr \II_\om^*(1) \CC_\om(1)
    +
    \f12\tr \II_\om^*(\s_z) \CC_\om(\s_z), 
    \label{eq-fom}
  \end{align}
  where we have used 
  \eqref{eq-trhs-2d} and
  $\NN^*(1)=1$, $\NN^*(\s_z)=\s_z$, 
  $\NN^*(\s_x)=\NN^*(\s_y)=0$, which follow from 
  \eqref{eq-strong-depha-noise}. 

  Let us write
  \begin{align}
    &\CC_\om(1)=1+\ve\a_\om\cdot\ve \s, 
    \q
    \CC_\om(\s_z)=\ve\b_\om\cdot \ve\s, 
    \\
    &
    \II^*_\om(1)=\g_{\om}+\ve\d_{\om}\cdot\ve\s, 
    \q
    \II^*_\om(\s_{z})=\ep_{\om}+\ve\z_{\om}\cdot\ve\s. 
  \end{align}
   Because $\nu +\ve{\xi}\cdot\ve{\sigma}\ge0$ holds 
   if and only if $\|\xi\|\le \nu$,
   positivity of $\CC_\om$ and 
   $\II_\om$ imply 
  $\norm{ \ve\a_\om\pm \ve\b_\om }\le1$ 
  and 
  $
  \norm{ \ve\d_\om\pm \ve\z_\om }
  \le
  \g_\om\pm\ep_\om
  $, respectively 
  (Consider the images of $1\pm\s_z$). 
  One therefore has
  \begin{align}
    &
    f_\om 
    =
    \g_\om+
     \ve\a_\om \cdot \ve\d_\om + \ve\b_\om \cdot \ve\z_\om
    \nn
    &
    =
    \g_\om+
    \f12
    \paren{
      \paren{\ve\a_\om+\ve\b_\om}
      \cdot
      \paren{\ve\d_\om+\ve\z_\om}
      +
      \paren{\ve\a_\om-\ve\b_\om}
      \cdot
      \paren{\ve\d_\om-\ve\z_\om}
    }
    \nn
    &
    \le
    \g_\om+
    \f12
    \paren{
      \norm{\ve\a_\om+\ve\b_\om}
      \norm{\ve\d_\om+\ve\z_\om}
      +
      \norm{\ve\a_\om-\ve\b_\om}
      \norm{\ve\d_\om-\ve\z_\om}
    }
    \nn
    &
    \le
    \g_\om+
    \f12
    \paren{
      (\g_\om+\ep_\om)
      +
      (\g_\om-\ep_\om)
    }
    \nn
    &=
    2\g_\om, 
\end{align}
where we have used the Cauchy-Schwarz inequality in the first
inequality. 
It follows from trace preservation of $\sum_\om\II_\om$, 
or $\sum_\om\II_\om^*(1)=1$, that 
$\sum_\om\g_\om=1$ holds. 
As a result, one obtains
  $
    \sum_\om f_\om\le2, 
  $
  which 
  is, 
  from Eq.~\eqref{eq:avefid2},
  equivalent to 
  \begin{align}
    \bar F\le \f23. 
  \end{align}
  This value  
  is attained by the both of the {\dn} and the {\dr} protocols, 
  which can be seen by 
  \eqref{eq-attain-dn} with $(d^1,d^2,d^3)=(1,0,0)$ and 
  by \eqref{eq-Fbar-DR-gendim}, respectively.


\begin{thebibliography}{99}
\bibitem{unitary}
  The optimal protocol $\brac{\paren{\II'_\om,\CC'_\om}}$ 
  for noise $\NN'$ of the form~\eqref{eq-nnuv} 
  is obtained 
  from 
  the optimal protocol $\brac{\paren{\II_\om,\CC_\om}}$ 
  for $\NN$
  by 
  setting $\II_\om':=\AA_{\da U}\circ\II_\om$ and 
  $\CC_\om':=\CC_\om\circ\AA_{\da V}$.
\end{thebibliography}
\end{document}